\documentclass[11pt,a4paper]{article}
\usepackage{coling2020}
\usepackage{times}  
\usepackage{helvet}  
\usepackage{courier}  
\usepackage{url}  
\usepackage{graphicx}  
\usepackage{latexsym}

\usepackage{microtype}

\colingfinalcopy 

\usepackage[symbol]{footmisc}
\usepackage{bm}
\usepackage{natbib}
\usepackage{amsmath,amsfonts,amssymb,amsthm}
\usepackage{tikz}
\usetikzlibrary{arrows, shapes, decorations.pathreplacing, calc, positioning}
\usepackage{tkz-graph}
\usepackage{algorithm}
\usepackage{algpseudocode}
\usepackage{caption}
\usepackage{subcaption}
\usepackage{fancyvrb}
\usepackage{enumitem,kantlipsum}
\usepackage{tabularx}
\usepackage{multirow}
\usepackage{graphics}
\usepackage{tablefootnote}
\usepackage{tikz-dependency}
\usepackage{bbm}

\usepackage[utf8]{inputenc}
\usepackage[english]{babel}

\colingfinalcopy 

\makeatother
\setlist{nolistsep}
\newtheorem{thm}{Theorem}

\newtheorem{lemma}[thm]{Lemma}
\newtheorem{prop}[thm]{Proposition}

\theoremstyle{definition}
\newtheorem{definition}{Definition}[section]

\title{Cross-Lingual Document Retrieval with Smooth Learning}

\author{Jiapeng Liu\thanks{\quad This work was majorly done while the first two authors were pursuing their Ph.D. degrees at Purdue University.} \and Xiao Zhang\footnotemark[1] \\
  Amazon.com \\
  {\tt \{liujiape, zhxao\}@amazon.com} \\
  \AND
  Dan Goldwasser \and Xiao Wang \\
  Purdue University \\
  {\tt \{dgoldwas, wangxiao\}@purdue.edu} \\}

\date{}

\begin{document}

\maketitle
\begin{abstract}
Cross-lingual document search is an information retrieval task in which the queries' language differs from the documents' language. In this paper, we study the instability of neural document search models and propose a novel end-to-end robust framework that achieves improved performance in cross-lingual search with different documents' languages. This framework includes a novel measure of the relevance, \textit{smooth cosine similarity}, between queries and documents, and a novel loss function, \textit{Smooth Ordinal Search Loss}, as the objective. We further provide theoretical guarantee on the generalization error bound for the proposed framework. We conduct experiments to compare our approach with other document search models, and observe significant gains under commonly used ranking metrics on the cross-lingual document retrieval task in a variety of languages.
\end{abstract}

\section{Introduction}
\label{sec:intro}
In modern search engines, cross-lingual information retrieval tasks are becoming prevalent and important. For example, when searching for products on a shopping website in English, immigrants tend to use their native languages to form the queries and would like to see the most desired products which are in English. Another example is in international trading, investors might use English to describe their product to search the sentiment in other languages on online forums from different countries, in order to understand the customers' attitude towards it. Despite that these tasks can be naturally formulated as information retrieval (IR) tasks and resolved by mono-lingual methods, the rising need of cross-lingual IR techniques also requires robust models to deal with queries and documents from different languages.

Early studies on documents retrieval mostly rely on lexical matching, which is error-prone in cross-lingual tasks as vocabularies and language styles usually change across different languages, as well the contextual information is largely lost. With the recent surge of deep neural networks (DNNs), researchers are able to go beyond lexical matching by building neural architectures to represent textual information of query and documents by vector representations via non-linear transformations, which have shown great successes in many applications  \citep{Salakhutdinov2009,Huang2013,Shen2014a,Palangi2016}.
 
Despite the success of those aforementioned advances, several existing difficulties have not been well addressed, e.g. \textit{exploding gradients} and convergence guarantee. The most widely used optimization method for training DNNs is stochastic gradient
descent, which updates model parameters by taking gradients of the weights. Meanwhile, cosine similarity, a commonly used measure of relevance between query and documents, is not stable as its gradient may go to infinity when the Euclidean norm of the representation vector is close to zero, leading to \textit{exploding gradients} and resulting in irrational training. Figure \ref{fig:grad} demonstrates the gradients of cosine similarity and our proposed smooth cosine similarity. In addition, in most of the previous works, the loss functions are either heuristically designed or margin-based \citep{rennie2005loss}, resolving to account for the particular property of the model but lacking interpretability and convergence guarantee. These new scenarios pose a new challenge towards the model: the robustness of the cross-lingual information process needs to be well addressed.

\begin{figure*}[!htb]
\centering
\includegraphics[width=0.6\textheight]{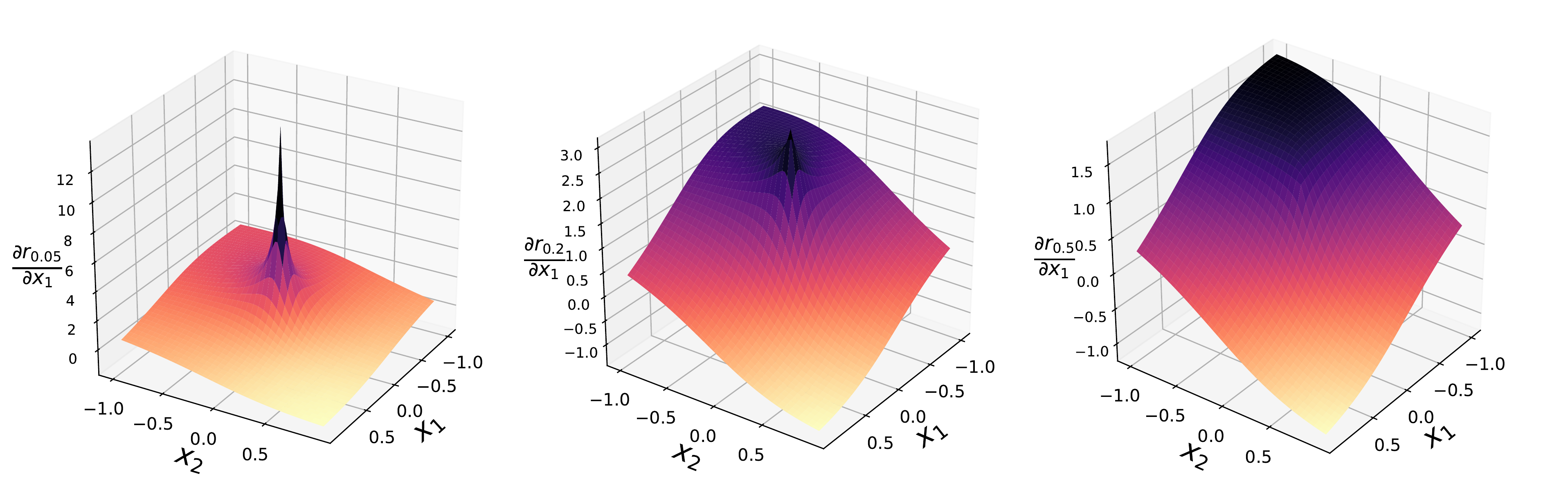}
\caption{This figure visualizes the gradient of proposed smooth cosine similarity with different $\epsilon$. We calculate the smooth cosine similarity $r_\epsilon$ of two vectors, ($x_1, x_2$) and (1, 1). The plots represent the heat-map of the partial derivative with respect to $x_1$, where $\epsilon$ is chosen as 0.05, 0.2, 0.5, from left to right. The plot of non-smooth case with $\epsilon=0$ is similar to the left plot, though the gradient goes to infinity at the center.}
\label{fig:grad}
\end{figure*}

To tackle these issues, we introduce an end-to-end robust framework that achieves high accuracy in the cross-lingual information retrieval task with different document languages. Particularly, for each query $\bm{q}$ from language $A$, we are given a set of documents $\bm{D} = (d_1, d_2, ..., d_m)$ from a different language $B$ and the degrees of \textit{relevance} $\bm{y} \in \mathbbm{N}^m$, where the entries $(y_1, y_2, ..., y_m)$ in $\bm{y}$ typically belong to ordered class $\{1, ..., K\}$. The goal of Learning-to-rank is to return an ordered or unordered, depending on the evaluation metrics, subset of the documents $\bm{D}$ that are more relevant to the query. Most common evaluation metrics, such as the Normalized Discounted Cumulative Gain (NDCG) and Precision, however, are discontinuous and thus cannot be directly optimized. As a result, researchers usually assume an unknown continuous relevance score $r$, where a higher value $r$ means a larger $y$ and higher relevance between the query and the document. To optimize the model, typically a \textit{surrogate} loss function is used, which is easier to optimize using the relevance score $r$. Subsequently a subset of the documents are being chosen by ordering $r$ from high to low, to serve the ranking purpose. 

Our contributions in this paper can be summarized as follows:
\begin{enumerate}[noitemsep, leftmargin=*]
\item First, we propose a novel measure of relevance between queries and documents, \textit{Smooth Cosine Similarity} (SCS), whose gradient is bounded such that \textit{exploding gradients} can be avoided, stabilizing the model training. 
%
\item Second, we propose a smooth loss function, \textit{Smooth Ordinal Search Loss} (SOSL), and provide theoretical guarantees on the generalization error bound for this proposed framework.   
\item Third, we empirically show significant gains with our approaches over other document search models
under commonly used ranking metrics on the cross-lingual document retrieval task, by conducting experiments in a variety of languages.
\end{enumerate}

\section{Related Works}
\label{sec:related}
\paragraph{Document Retrieval} 
Researchers have applied machine learning methods to a variety of document retrieval tasks. \citet{deerwester1990indexing} proposed LSI that maps a query and its relevant documents into the same semantic space where they are close by grouping terms appearing in similar contexts into the same cluster. \citet{Salakhutdinov2009} proposed a semantic hashing (SH) method which used non-linear deep neural networks to learn features for information retrieval.

Siamese Neural Networks was first introduced by \citet{bromley1993signature}, where two identical neural architectures receive different types of input vectors (e.g., query and document vectors in information retrieval tasks).  \citet{Huang2013} introduced deep structured semantic models (DSSM) which projected query and document into a common low-dimensional space using feed-forward neural network models, where they chose the cosine similarity between a query vector and a document vector as their relevance score. \citet{Shen2014a} and \citet{Palangi2016} extended the feed-forward structure in \textit{DSSM} to convolutional neural networks (CNN) and recurrent neural networks (RNN) with Long Short-Term Memory (LSTM) cells. Differing from previous works that used click-through data with only two classes,
\citet{Nigam2019} proposed a loss function that differentiates three classes (relevant, partially relevant, and irrelevant) in document search with huge amount of real commercial data.


\paragraph{Cross-Lingual Information Retrieval} Traditionally, Cross-Lingual Information Retrieval (CLIR) is conducted in two steps in a pipeline: machine translation followed with monolingual information retrieval \citep{nie2010cross}. However, this approach requires a well-trained translation model and usually suffers from translation ambiguity \citep{zhou2012translation}. The error propagation from machine translation may even deteriorate the retrieval results. As an alternative, pre-trained word embeddings \citep{Mikolov2013,pennington-etal-2014-glove} have led to a surge of improved performance on many language tasks, which learns word representations of different languages on large scale text corpora. Nevertheless, the training objective of these embeddings are different from IR tasks, thus their direct application may be limited. 


\paragraph{Generalization Error}  There has been previous works on the generalization error of learn-to-rank models. \citet{lan2008query} analyzed the stability of pairwise models and gave query-level generalization error bounds. \citet{lan2009generalization} provided a theoretical framework
for ranking algorithms and proved generalization error bounds for three list-wise losses: \textit{ListMLE}, \textit{ListNet} and \textit{RankCosine}. \citet{chapelle2010gradient} introduced annealing procedure to find the optimal smooth factor on an extension of surrogate loss \textit{SoftRank}, and derived a generally applicable bound on the generalization error of query-level learning-to-rank algorithms. \citet{tewari2015generalization} proved that several loss functions used in learning-to-rank, such as cross-entropy loss, have no degradation in generalization ability as document lists become longer. However, theoretical analysis toward search models with neural architectures is still very limited.





\section{Smooth Neural Document Retrieval}
\label{sec:method}
In this section, we propose a novel Smooth Cross-Lingual Document Retrieval framework. This framework consists of three parts. First, we use neural models to encode queries and documents from different languages, and represent them by low-dimensional vectors. Second, we propose a smooth cosine similarity to indicate the relevance score $r$, which avoids gradient explosion and therefore stabilizes the training process. Finally, we introduce Smooth Ordinal Search Loss for optimizing $r$.  

 
\subsection{Text Representation}

We use embeddings as low dimensional dense vectors to represent both queries and documents.
Differing from mono-lingual tasks, in cross-lingual document retrieval one can rarely observe common tokens between queries and documents. Therefore, cross-lingual document retrieval usually requires embeddings built from different vocabularies as queries and documents are from two different languages \citep{Sasaki2018}. This requirement naturally expands the size of parameters of the model if we regard embeddings as part of model parameters and fine-tune them during training, thus the retrieval model itself demands a higher stability.

Queries and documents are represented numerically. A query $\bm{q}$ is tokenized into a list of words $\bm{q}=q_1q_2...q_{l_q}$ of length $l_q$. For example, the tokenization of ``Apple is a fruit'' is [``Apple'', ``is'', ``a'', ``fruit''] of length 4. Similarly, $\bm{d}=d_1d_2...d_{l_d}$ is an expression of a document of length $l_d$ from document language. For simplicity, we use A and B to represent query language and document language, respectively. We also choose the same dimension $p$ for both word embeddings $Emb_A$ and $Emb_B$. Therefore, we encode the query $\bm{q}$ and the document $\bm{d}$ by embedding representations $Q \in \mathbbm{R}^{p \times l_q}$ and $D \in \mathbbm{R}^{p \times l_d}$, where the \textit{i}th column of $Q$ is the word embedding from $Emb_A$ for the token in the \textit{i}th position in the query, and the \textit{j}th column of $D$ is the word embedding from $Emb_B$ for the token in the \textit{j}th position in the document. 

\subsection{Neural Model Architecture}

With the embedding representations $Q$ and $D$, which can be of different sizes due to different number of tokens in the query and the document, we apply neural models onto them to obtain vectors of the same size. Carefully designed neural models are able to perform dimension reduction regarding the sequence length, and thus project the raw texts into the same Euclidean space of the same dimension for both the queries and the documents, regardless of the number of tokens in them. Thus, we can quantify the relevance between a query and a document as they are in the same space by calculating standard metrics, such as cosine similarity. For example, if the embedding size is $p$, then a query of length $l_q$ can be represented by $Q \in \mathbbm{R}^{p \times l_q}$, and the representations of two documents of length $l_{d1}$ and $l_{d2}$ are $D_1 \in \mathbbm{R}^{ p \times l_{d1}}$ and $D_2 \in \mathbbm{R}^{p \times l_{d2}}$, respectively. In such case, $Q$, $D_1$ and $D_2$ are in different space. To mitigate this difficulty that identifies quantitative comparison of the relevance of two pairs, $(Q, D_1)$ and $(Q, D_2)$, a well designed neural model is used to transform both $Q$ and $D$ into $\mathbbm{R}^{p}$.

In order to achieve this goal, we propose to use average pooling over the columns of $Q$ and $D$, followed by a non-linear activation function \textit{tanh} for both query and document models. For a query $\bm{q}$, the query model $f_q: Q \rightarrow \bm{v}_q \in \mathbbm{R}^p$ takes embedding $Q$ as input and outputs a final representation $\bm{v}_q$. Similarly, the document model $f_d: D \rightarrow \bm{v}_d \in \mathbbm{R}^p$ also generates a final representation $\bm{v}_d$ into the same space using $D$. There are other modeling choices such as LSTM and CNN. We compare the performance of different neural models in the experiments.

The benefits of this model are in two folds: on one hand, despite non-linearity, the models are smooth and Lipschitz continuous with respect to embedding parameters, and therefore benefits from convergence of generalization error during training; on the other hand, using average pooling avoids extra parameters in the model, thus simplifies the model space and reduces tuning efforts during training, in accordance to the findings in \citet{Nigam2019}'s work.

\subsection{Smooth Cosine Similarity}
Cosine similarity has been widely used to find the relevance between queries and documents in information retrieval. Given two vectors $\bm{x}$ and $\bm{z}$ of same size in Euclidean space, cosine similarity $\textit{cos}(\bm{x}, \bm{z}) = \frac{\bm{x}^{\intercal}\bm{z}}{\|\bm{x}\|\|\bm{z}\|}$ measures how similar these two vectors are irrespective of their norms, i.e. $\|\bm{x}\|$ and $\|\bm{z}\|$. However, the norms of the vectors play crucial roles when calculating the gradient. More specifically, the gradient goes to infinity if the norms are close to zero, and results in unstable weights update during training. This phenomenon is also known as \textit{exploding gradients}. The intuition is that for a vector of small norm, a slight disturbance can greatly change the angle between itself and another vector, i.e. the cosine similarity. The use of cosine similarity can lead to {exploding gradients} regardless of the model structure when gradient descent methods are used for optimization. Recently, most semantic matching models and learning-to-rank models are constructed based on neural architectures. Thus, these retrieval models suffer from this issue greatly since commonly they are optimized by gradient descent methods.

To increase the stability of model training, we further propose \textit{Smooth Cosine Similarity} (SCS) in replace of the regular cosine similarity. We define the SCS between a query and a document as 
\begin{align}
    r_\epsilon(\bm{q}, \bm{d}) = \frac{f_q(\bm{q})^T f_d(\bm{d})}{(\|f_q(\bm{q})\|+\epsilon)(\|f_d(\bm{d})\|+\epsilon)},
\end{align}
where $\|\cdot\|$ is the Euclidean norm and $\epsilon > 0$. Under the framework of SCS, the gradient of $r_\epsilon$ with respect to $f_q(\bm{q})$ and $f_d(\bm{d})$ is upper bounded in the whole space and thus stabilizes training procedure. Moreover, by introducing this additional smoothness hyper-parameter into the norm of the feature representation vectors, the similarity score not only measures the angle between vectors, but also adds information about the norm of the vectors. As a result, SCS is not order-preserving from cosine similarity, i.e., $cos(\bm{x}, \bm{z}_1) > cos(\bm{x}, \bm{z}_2)$ does not necessarily imply $r_\epsilon(\bm{x}, \bm{z}_1) > r_\epsilon(\bm{x}, \bm{z}_2)$. 
The choice of the hyper-parameter $\epsilon$ is flexible and not sensitive to the model performance, which is further analyzed in the experiments.

Another common method to avoid exploding gradients is \textit{gradient clipping}, i.e., clipping gradients if their norm exceeds a given threshold. Our proposed SCS does not exclude gradient clipping and in fact they complement with each other. In our pilot experiments, we observe merely using gradient clipping is not sufficient in our cross-lingual document retrieval setup. By adding SCS, we observe improved performance over using gradient clipping alone.

\subsection{Smooth Ordinal Search Loss}
In a search ranking model, it is critical to define a \textit{surrogate} loss function since the ranking metrics, such as \textit{NDCG} and \textit{Precision} are not continuous therefore difficult to optimize. On the choice of a \textit{proper} loss function, one has to consider two criteria: first, minimizing the \textit{surrogate} loss on training set should imply a small \textit{surrogate} loss on test set; second, a small \textit{surrogate} loss on test set should imply desired ranking metrics results on test set  \citep{chapelle2011future}. To deal with the first criterion, we formulate the search ranking model as an ordinal regression problem. For the second one, we propose to use \textit{Smooth Ordinal Search Loss} (SOSL) as the \textit{surrogate} loss.

Recall that a pair of query and document has an ordered relevance level, and the goal of a search ranking model is to select a subset of documents such that more relevant documents are ranked on the top while less relevant documents are ranked lower. Taking three-class ranking problem as a concrete example, the pairs can be grouped into \textit{relevant, partially relevant} and \textit{irrelevant}. If a mis-ranking has to exist for a \textit{relevant} pair, it is preferred to rank a \textit{partially relevant} document over an \textit{irrelevant} document, which means not all mistakes are equal. 

%

\begin{figure}[!htb]
\centering
\includegraphics[width=0.55\textheight]{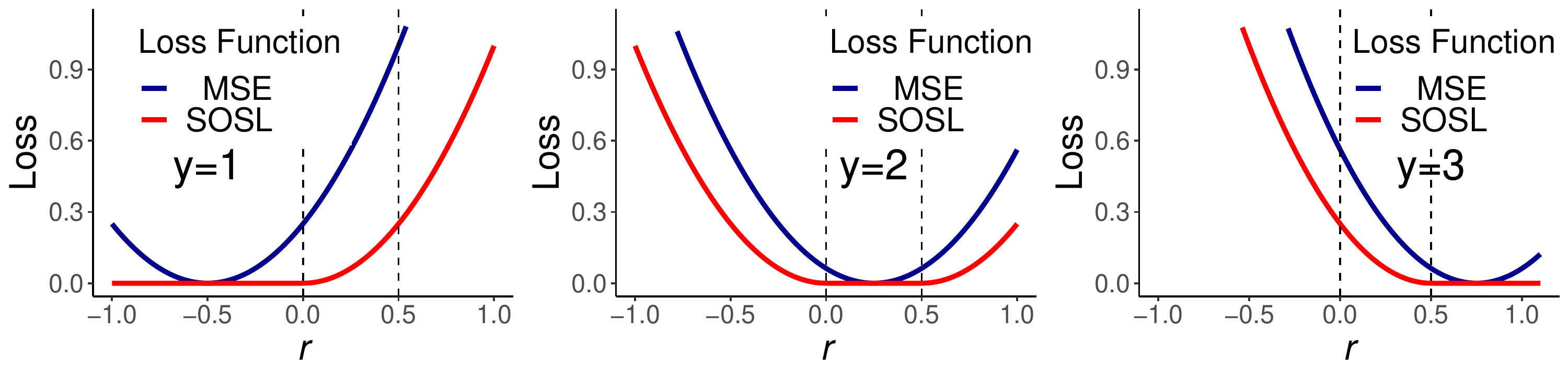}
\caption{These figures illustrate the difference between SOSL and MSE for a three-class $\{1, 2, 3\}$ ordinal regression task where the true label $y$ varies from 1 to 3, with $(\theta_1, \theta_2) = (0, 0.5)$.}
\label{fig:loss}
\end{figure}

Inspired by the immediate threshold with margin penalty function construction \citep{rennie2005loss}, we propose our loss function \textit{Smooth Ordinal Search Loss} (SOSL) formally as following:
\begin{align}
    SOSL(r_\epsilon, y) &= (\theta_y - r_\epsilon)^2\mathbbm{I}[r_\epsilon > \theta_y] 
    + (r_\epsilon - \theta_{y-1})^2\mathbbm{I}[r_\epsilon < \theta_{y-1}],
    \label{eq:osl}
\end{align} 
where $r_\epsilon$ is the smooth relevance score between a query $\bm{q}$ and a document $\bm{d}$, $y \in \{1, ..., K\}$ is the ordered class label denoting the general relevance degree, and $-1 = \theta_0 < ... < \theta_K = 1$ are the thresholds and $\mathbbm{I}$ is the indicator function. Differing from that of the margin penalty function
used in \citet{rennie2005loss}'s work, we choose smooth function $(\cdot)^2$. The thresholds are within the range $[-1, 1]$ in our setup, instead of the whole real line, due to the property of smooth cosine similarity.

The interpretation of this loss function is intuitive, as shown in Figure \ref{fig:loss}: if the relevance score falls into the correct segmentation (the true ordered class of a particular pair of query and document is $y'$ and $\theta_{y'-1} \leq r' \leq \theta_{y'}$), then the loss is $0$; otherwise the loss is the degree of the relevance violating the threshold.

\section{Theoretical Analysis}
\label{sec:theory}
Generalization error measures the difference between training error and testing error. Typically, if the generalization error of a model is bounded and converges to zero, then minimizing empirical loss on training set implies that the expected loss on unseen testing set is also minimized. Although previous works gave generalization bounds for \textit{surrogate} losses in learning-to-rank model \citep{lan2008query, lan2009generalization, tewari2015generalization}, to the best of our knowledge, no theoretical result has been derived in search models with \textit{neural architecture}. Here, we prove a generalization error bound for the commonly used SGD procedure. This error bound suggests that the generalization gap at any training step $T$ converges to zero when the number of training pairs of query and document goes to infinity. We show the detailed proof in the appendix. 

The following Proposition \ref{prop:3} and Lemma \ref{lemma:smooth} show that SOSL is both smooth and Lipschitz continuous with respect to not only relevance score $r$ but also embedding parameters.  
\begin{prop}\label{prop:3}
SOSL is smooth and Lipschitz continuous with respect to $r$.
\end{prop}
\begin{lemma}\label{lemma:smooth}
Let $l(r,y)$ be a smooth and Lipschitz continuous loss function with respect to $r$, then in our search ranking model, $l(r_\epsilon(\bm{q}, \bm{d}), y) = l(f_q, f_d, y)$ is also smooth and Lipschitz continuous with respect to model parameters, i.e., embeddings.
\end{lemma}

From Lemma \ref{lemma:smooth}, we can assume $l(f_q, f_d, y)$ to be $\alpha$-Lipschitz and $\gamma$-smooth.  Next, suppose $S=((\bm{q}_1,\bm{d}_1,y_1), ..., (\bm{q}_n, \bm{d}_n, y_n))$ is a training set sampled from the data distribution with sample size $n$, and assume unseen testing data is from distribution $E$. Let $f = (f_q, f_d)$ be the neural models for queries and documents. By defining $L_S^l(f)$ as the mean training error and $L_E^l(f)$ the expected error on testing set, the following Theorem establishes the upper bound for generalization error $L^l_E(f) - L^l_S(f)$.

\begin{thm}(Generalization Error bound) \label{thm}
Let $l$ be a smooth loss function bounded by $M$. Suppose we run SGD for $T steps$ (with step size $\alpha_t = c/t, t \in \{1,2,\dots, T\}$), with probability at least $1- \delta$ over the draw of $S$, 
\begin{align*}
    & L^l_E(f) - L^l_S(f)
     \leq 2\beta(n) + (4n\beta(n) + M) \sqrt{\frac{\ln 1/ \delta}{2n}},
\end{align*}
where $\beta(n) \lessapprox \frac{T^{1-1/\gamma c}}{n}$.
\end{thm}
Note that in Theorem \ref{thm}, the bound does not depend on thresholds $\theta$. We ignore the dependency on $\epsilon$ and $\alpha$ as both are constants. Theorem \ref{thm} suggests that for any $T$, the generalization error after $T$ steps for SGD converges as the sample size of $S$ increases. For some proper step size, we can allow $T$ to increase with $n$. In particular, if $\gamma c < 2$, then $T$ may increase at a rate of $n$ while the generalization error still converges.


\section{Experiments}
\label{sec:exp}

\paragraph{Datasets}
We use the publicly available large-scale Cross-Lingual Information Retrieval (CLIR) dataset from Wikipedia \citep{Sasaki2018} for our experiments. All the queries are in English, extracted as the first sentences from English pages, with title words removed. The \textit{Relevant} (MR) documents are the foreign-language pages having inter-language link to the English pages; the \textit{Partially Relevant} (SR) documents are those having mutual links to and from the relevant documents. Additionally, we randomly sample 40 other pages as \textit{Irrelevant} (NR) documents for each query. 
To provide a comprehensive study, we use the document datasets of two high-resource languages: French (fr), Italian (it) and two low-resource languages: Swahili (sw), Tagalog (tl). Queries are randomly split into training, validation and testing set with the rate of 3:1:1. We include the data statistics in Table \ref{table:data}.

\begin{table}[!htb]
	\centering
	{\begin{tabular}{cccc} 
		\hline
		Language & \#Query &  \#SR/Q & \#Documents\\ \cline{1-4}
		French  &	25000 & 12.6 &  1894000\\ 
		Italian	&  25000&   11.7 & 1347000 \\  	
    	Swahili		&22793 &  1.5 & 37000 \\  
		Tagalog		& 25000 &  0.6 & 79000  \\  \cline{1-4}	
	\end{tabular}}
		\caption{We show the number of queries, the average number of Partially Relevant (SR) documents per query and the total number of documents of each language. }
\label{table:data}
\end{table}	



\paragraph{Evaluation Metrics}  
We apply the commonly used ranking metrics in our experiments for evaluation, including Precision, NDCG, MAP and MRR. For each query, the corresponding documents are sorted by their relevance score, and the metrics are averaged over all queries. $P_{mr}@1$ represents the precision for MR document at top 1 position. $P_r@5$ represents the precision for MR and SR documents combined at top 5 positions. NDCG@5 is Normalized Discounted Cumulative Gain for top 5 documents. MAP is calculated as the mean of the average precision scores for each query. $MRR_{mr}$ is the Mean reciprocal rank of MR document.

\paragraph{Experimental Setup}

In our experiments, we use the pre-trained \textit{Polyglot} \citep{polyglot:2013:ACL-CoNLL} embeddings of dimension 64 as the initialization for the corresponding languages. 
These embeddings are fine-tuned during the training. We also observe improved performance by shuffling the training set.
We select the thresholds $\bm{\theta} = (0.2, 0.7)$ for SOSL, via cross-validation using grid searching. We use Adam optimizer in all experiments for optimization, with fixed learning rate 0.01. We set batch sizes as 128 and we stop after 30 epochs for all languages, resulting in about 310,000 training steps on French and 240,000 on Tagalog, while the other two languages' training steps are in-between. We release our code and experimental setup to benefit the community and promote researches along this topic\footnote{https://github.com/JiapengL/multi\_ling\_search}.



\renewcommand\arraystretch{1.1}	
\begin{table*}[!htb]
	\centering
	\scalebox{0.8}{
	{\begin{tabular}{ccccc ccccc} 
			\hline
			Language & Loss function &$P_{mr}@1$ & $P_{mr}@5$ & $P_r@5$ & $NDCG@5$ & MAP & $MRR_{mr}$ & $MRR_{r}$\\ \cline{1-9}
			\multirow{4}*{French}
			& $SOSL$ & $\mathbf{0.438}$ & $\mathbf{0.832}$ & $\mathbf{0.607}$ & $\mathbf{0.811}$ & $\mathbf{0.841}$ & $\mathbf{0.607}$ & $\mathbf{0.919}$  \\
            & $3Part_{l2}$ & 0.411 & 0.763 & 0.560 & 0.754 & 0.766 & 0.565 & 0.889\\	
            & MSE  &0.253 & 0.700 & 0.603 & 0.727 & 0.792 & 0.443 & 0.854\\
            & $PO_S$ & 0.254 & 0.704 & 0.604 & 0.729 & 0.795 & 0.445 & 0.856  \\  
			 \cline{1-9}

			\multirow{3}*{Italian}
			& $SOSL$& $\mathbf{0.401}$ & $\mathbf{0.791}$ & $\mathbf{0.614}$ & $\mathbf{0.798}$ & $\mathbf{0.838}$ & $\mathbf{0.568}$ & $\mathbf{0.912}$  \\ 
            & $3Part_{l2}$ & 0.385 & 0.748 & 0.572 & 0.751 & 0.768 & 0.545 & 0.883\\
            & MSE & 0.231 & 0.699 & 0.618 & 0.731 & 0.803 & 0.427 & 0.862\\	
            & $PO_S$& 0.232 & 0.705 & 0.619 & 0.734 & 0.806 & 0.430 & 0.863  \\  
			 \cline{1-9}	

			\multirow{3}*{Swahili}
			& $SOSL$ &  $\mathbf{0.600}$ & $\mathbf{0.916}$ & $\mathbf{0.341}$ & $\mathbf{0.793}$ & $\mathbf{0.776}$ & $\mathbf{0.735}$ & $\mathbf{0.833}$  \\ 	
            & $3Part_{l2}$ & 0.599 & 0.907 & 0.314 & 0.771 & 0.737 & 0.732 & 0.827\\
            & MSE & 0.351 & 0.851 & 0.360 & 0.724 & 0.738 & 0.558 & 0.771\\
            & $PO_S$ &  0.351 & 0.855 & 0.362 & 0.726 & 0.740 & 0.560 & 0.772 \\  
			 \cline{1-9}	

			\multirow{3}*{Tagalog}
			& $SOSL$ & $\mathbf{0.596}$ & $\mathbf{0.928}$ & $\mathbf{0.254}$ & $\mathbf{0.799}$ & $\mathbf{0.772}$ & $\mathbf{0.738}$ & $\mathbf{0.794}$  \\ 	           
			& $3Part_{l2}$ & 0.583 & 0.925 & 0.251 & 0.790 & 0.760 & 0.730 & 0.785\\
            & MSE & 0.455 & 0.892 & 0.251 & 0.732 & 0.709 & 0.639 & 0.726\\	
            & $PO_S$ & 0.463 & 0.896 & 0.252 & 0.737 & 0.715 & 0.645 & 0.730  \\  
			 \cline{1-9}				
			
	\end{tabular}}
	}
		\caption{This table shows results with different losses in different languages. The $\epsilon$ is 1 for all the losses.  }
\label{table:loss}
\end{table*}



\subsection{Results} \label{sec:res}

\paragraph{SOSL and Other Loss Functions}

We first present the results of using different loss functions in different languages with smooth cosine similarity in Table \ref{table:loss}. We compared with three commonly used loss: Mean Square Error(MSE), Proportional Odds Loss (PO) \citep{mccullagh1980regression} and $3Part_{L2}$ \citep{Nigam2019}. The hyperparameter $\epsilon$ is fixed to 1 for all losses. We observe that SOSL outperforms  other loss functions in all languages and all metrics. We attribute this success to two folds: first, SOSL encourages smoothness over the optimization of parameters and thus guarantees convergence of generalization error; second, SOSL adds no penalty on the loss when the relevance score falls into the correct segmentation. Note that the performance in low-resource languages (\textit{sw} and \textit{tl}) is better
 than high-resource languages (\textit{fr} and \textit{it}) for some metrics, because the former languages have fewer SR documents and fewer total number of documents for each query, which makes it easier to distinguish MR and NR documents and is therefore more likely to rank the MR document on the top. 

To investigate the reason why SOSL performs the best, we plot the density curves of the relevance score predicted by each loss function on the training set, and illustrate them in Figure \ref{fig:plot}. All three classes are well separated in top left with SOSL. For 3partL2 loss, however, we can notice a large overlap in group NR and SR. In the plots of $PO$ and MSE losses, few MR documents are classified correctly. Besides, MR and SR documents are mixed together, with a portion of MR incorrectly classified as NR, which is unpreferable. This analysis showcases the power of SOSL to distinguish different types of documents at the end of the training.

\begin{figure}[!htb]
\centering
\includegraphics[height=0.15\textheight]{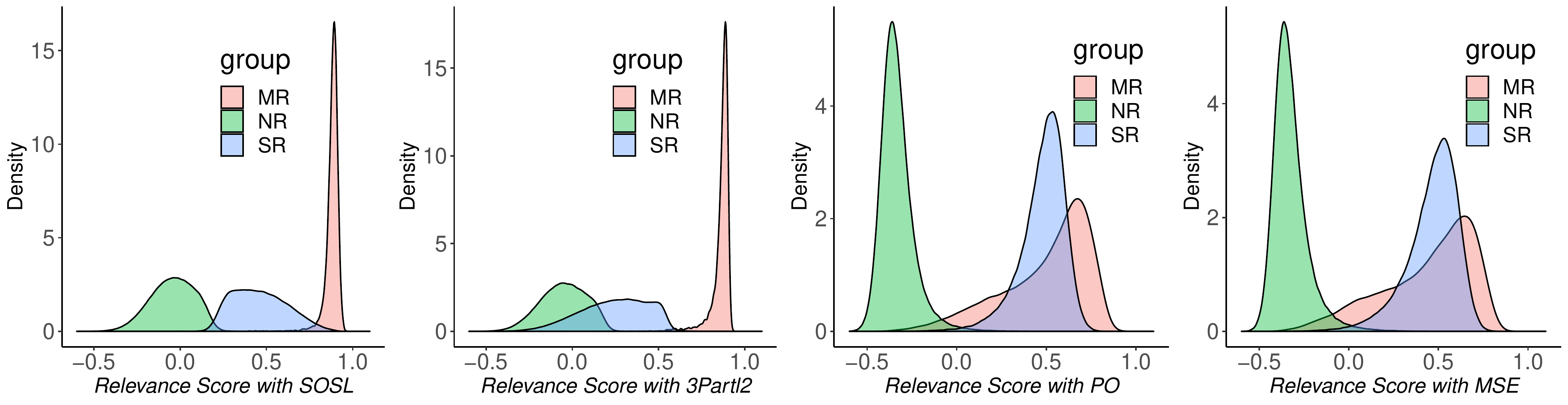}
\caption{The plots show the density of relevance score distribution with each loss on training set. In all losses $\epsilon = 1$ is applied.}
\label{fig:plot}
\end{figure}

\paragraph{Average Pooling and Other Neural Architectures}

We also compare the average pooling with other popular neural architectures. We choose DSSM-CNN model \citep{Shen2014a} and the DSSM-LSTM model \citep{Palangi2016}, which are widely used in information retrieval, to compare with. After tuning these two models, we specify their hyper-parameters as follows: we set 0.001 as the initial learning rate and 0.95 as the exponential decay rate for DSSM-CNN and DSSM-LSTM. We set the batch size 128 for DSSM-CNN and stop after 30 epochs and set the batch size 64 for DSSM-LSTM and stop after 15 epochs.

We set the window size as 3, i.e., word-3-gram for DSSM-CNN, with 300 filters in the convolutional layer. A maximum pooling layer and a fully connected layer with output size 64 are stacked after the convolutional layer. In the DSSM-LSTM model, we use the bidirectional LSTM model with the hidden units 64, and concatenate the first hidden state for the forward LSTM and the last hidden state for the backward LSTM, as the output of LSTM layer. This is then followed by a fully connected layer with output size 64 as the final output. To regularize the DSSM-CNN and the DSSM-LSTM model, we apply dropout \citep{srivastava2014dropout} on the word embedding layer with 0.4 dropout rate. We observe using separate modules for queries and documents can improve the performance comparing to sharing the same model, so two modules with same structure for queries and documents are used for CNN and LSTM.


In Table \ref{table:structure}, we compare the results of the average pooling architecture with the DSSM-CNN model and the DSSM-LSTM model. Both the two original studies used the data type which only contains two classes, positive and negative, therefore they designed the binary loss function accordingly. For a fair comparison, all the neural architectures are followed by the same SOSL with the same thresholds $\bm{\theta}$ used in Table \ref{table:loss}. In all the languages and all the evaluation metrics, average pooling performs best among the three. We attribute the success of average pooling to two reasons: first, queries typically have fewer words than documents, and do not tend to have long-range dependencies; second, the original approaches only deal with two classes, lacking the flexibility of dealing with more classes. In addition, CNN and LSTM are more difficult to optimize, as they require more computational resources, training time, and tend to overfit. Our results are in accordance with the findings in \citet{Nigam2019}'s study.

\begin{table*}[!htb]
	\centering
	\scalebox{0.8}{

	{\begin{tabular}{ccccc ccccc} 
			\hline
			Language & Loss function &$P_{mr}@1$ & $P_{mr}@5$ & $P_r@5$ & $NDCG@5$ & MAP & $MRR_{mr}$ & $MRR_{r}$\\ \cline{1-9}
			\multirow{4}*{French}
			& Average Pooling & $\mathbf{0.438}$ & $\mathbf{0.832}$ & $\mathbf{0.607}$ & $\mathbf{0.811}$ & $\mathbf{0.841}$ & $\mathbf{0.607}$ & $\mathbf{0.919}$   \\ 
			& DSSM-CNN& 0.262 & 0.656 & 0.542 & 0.570 & 0.709 & 0.437 & 0.812  \\ 		
			& DSSM-LSTM &  0.335 & 0.718 & 0.560 & 0.716 & 0.748 & 0.503 & 0.846\\    \cline{1-9}
			
            \multirow{4}*{Italian}
			& Average Pooling & $\mathbf{0.401}$ & $\mathbf{0.791}$ & $\mathbf{0.614}$ & $\mathbf{0.798}$ & $\mathbf{0.838}$ & $\mathbf{0.568}$ & $\mathbf{0.912}$  \\  
			& DSSM-CNN& 0.217 & 0.620 & 0.547 & 0.656 & 0.709 & 0.394 & 0.805  \\ 		
			& DSSM-LSTM &  0.248 & 0.675 & 0.551 & 0.675 & 0.714 & 0.434 & 0.814  \\   \cline{1-9}
			
			\multirow{4}*{Tagalog}
			& Average Pooling & $\mathbf{0.596}$ & $\mathbf{0.928}$ & $\mathbf{0.254}$ & $\mathbf{0.799}$ & $\mathbf{0.772}$ & $\mathbf{0.738}$ & $\mathbf{0.794}$ \\ 
			& DSSM-CNN& 0.450 & 0.861 & 0.235 & 0.699 & 0.668 & 0.624 & 0.693  \\ 		
			& DSSM-LSTM &  0.526 & 0.907 & 0.242 & 0.748 & 0.709 & 0.688 & 0.736  \\    \cline{1-9}
            \multirow{4}*{Swahili}
			& Average Pooling & $\mathbf{0.600}$ & $\mathbf{0.916}$ & $\mathbf{0.341}$ & $\mathbf{0.793}$ & $\mathbf{0.776}$ & $\mathbf{0.735}$ & $\mathbf{0.833}$  \\ 
			& DSSM-CNN& 0.457 & 0.869 & 0.310 & 0.700 &  0.680 & 0.628 & 0.743  \\ 		
			& DSSM-LSTM & 0.527 & 0.897 & 0.319 & 0.738 & 0.713 & 0.685 & 0.772   \\    \cline{1-9}		
	\end{tabular}}
	}
		\caption{This table shows the results on different model structures in different document languages.}
\label{table:structure}
\end{table*}

\subsection{Impact of Hyper-parameters}
Smoothness factor $\epsilon$ determines the smoothness of the cosine similarity. When $\epsilon$ is 0, the loss function with respect to weights in \textit{neural} models is non-smooth and thus generalization cannot be guaranteed. To show the improvement of adding smoothness and illustrate the effect of $\epsilon$, 
we vary $\epsilon$ from 0 to 2 for SOSL, in document language French. For each $\epsilon$, we use grid search to find the best $\theta$. The results are visualized in Figure \ref{fig:ep}. It is observed that model performance is improved when the model becomes smooth. Besides, within a large range of $\epsilon$ (from 0.25 to 1.75), adding smoothness factor surpasses non-smooth cosine similarity. We see a concave curve in all metrics but the choice of $\epsilon$ is relatively non-sensible. We suggest taking $\epsilon$ within $[0.25, 1]$ for desired output and stability. 
\begin{figure}[!htb]
\centering
\includegraphics[height=0.1\textheight]{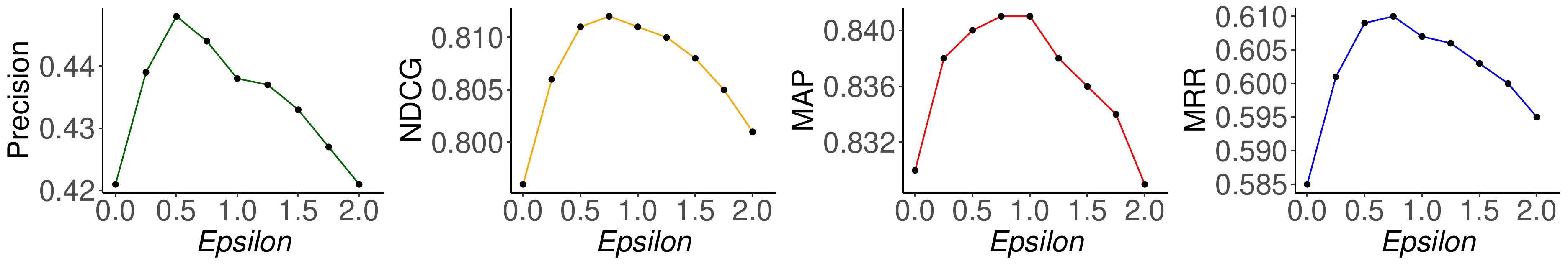}
\caption{These figures visualize the performance for SOSL with different $\epsilon$ in French. For each $\epsilon$, thresholds $\theta$ are chosen by grid search with the best results. Note that $\epsilon = 0$ is the case for non-smoothness. }
\label{fig:ep}
\end{figure}

\subsection{Different number of negative samples} \label{sec:neg}
In real industrial searching system, ranking is usually run after documents ``filtering'' process, e.g., the matching stage, which could greatly reduces the number of documents to be ranked. In Figure \ref{fig:neg}, we explore the effect of different number of \textit{irrelevant} (NR) documents. We create 9 different datasets where the number of NR documents per query varies from 20 to 100. The document language is the high-resource language French. We sampled same number of queries for training, validating and testing datasets as the experiments discussed earlier in this paper. The average number of SR documents per query varies but is still close to 12.6. The red curve indicates, as the number of NR documents increases, the data is more ``noisy'' and therefore more difficult to correctly rank and predict MR document. On the other hand, NDCG, MRR and MAP are relatively resistant to the increased number of NR documents since they only decrease about 10\% while the number of NR documents is $5\times$ more. This also validates the stability of our proposed framework.

\begin{figure}[!htb]
\centering
\includegraphics[width=0.5\textwidth]{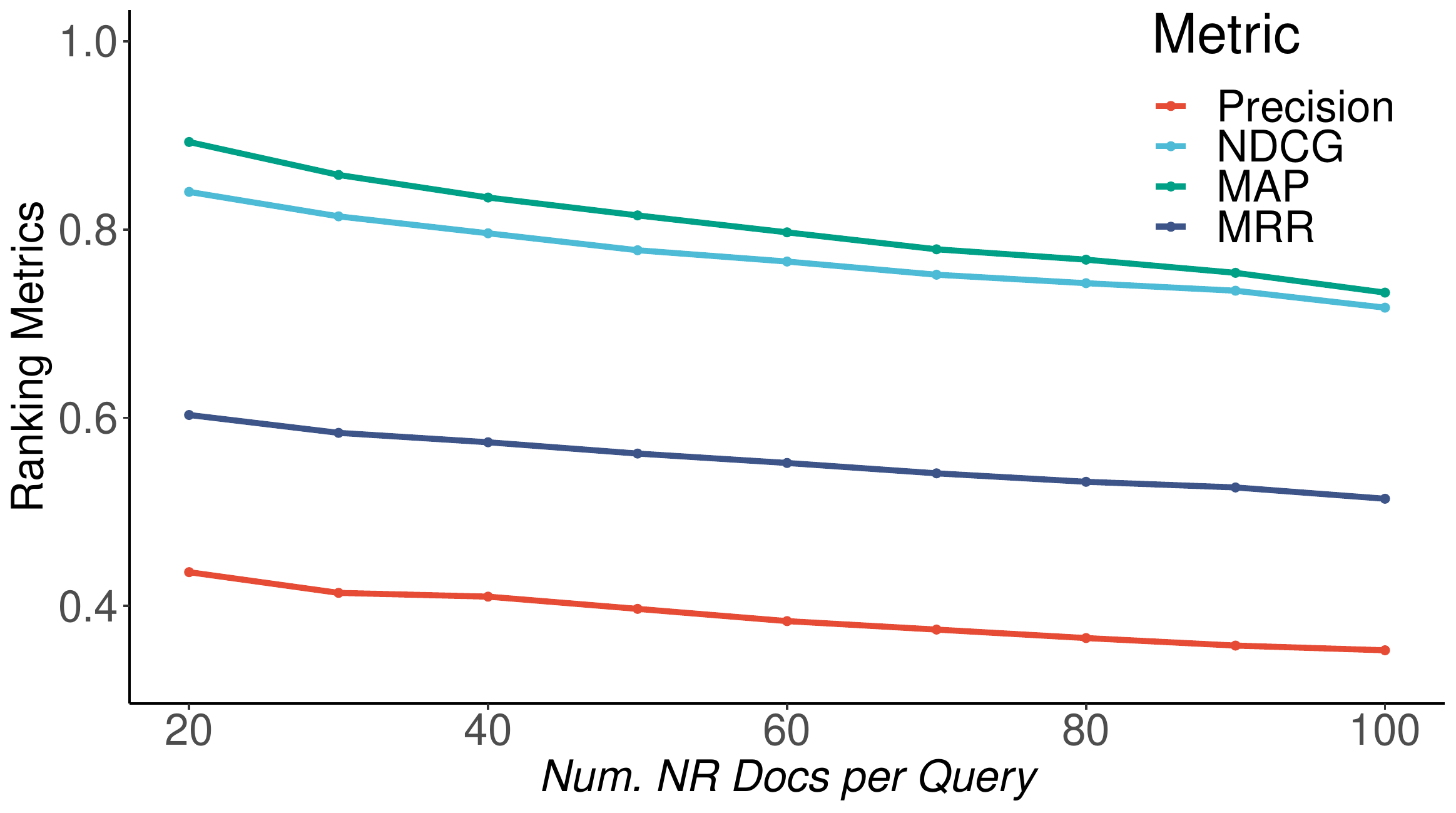}
\caption{The figure shows the relationship between evaluation metrics and the number of negative samples for each query in French. Precision and MRR representing $P_{mr}@1$ and $MRR_{mr}$ are discussed in Table \ref{table:loss}, respectively; NDCG is calculated from top 5 documents.}
\label{fig:neg}
\end{figure}

\section{Conclusion}
\label{sec:concl}
In this study, we propose a smooth learning framework for cross-lingual information retrieval task. We first suggest a novel measure of relevance between queries and documents, namely \textit{Smooth Cosine Similarity} (SCS), whose gradient is bounded thus able to avoid \textit{exploding gradients}, enforcing the model to be trained in a more stable way. Additionally, we propose a smooth loss function: \textit{Smooth Ordinal Search Loss} (SOSL), and provide theoretical guarantees on the generalization error bound for the whole proposed framework. Further, we conduct intensive experiments to compare our approach with existing document search models, and show significant improvements with commonly used ranking metrics on the cross-lingual document retrieval task in several languages. Both the theoretical and the empirical results imply the potentially wide application of this smooth learning framework.

\bibliographystyle{named}

\bibliography{mybib}

\begin{thebibliography}{}

\bibitem[\protect\citeauthoryear{Agarwal}{2008}]{agarwal2008generalization}
Shivani Agarwal.
\newblock Generalization bounds for some ordinal regression algorithms.
\newblock In {\em International Conference on Algorithmic Learning Theory},
  pages 7--21. Springer, 2008.

\bibitem[\protect\citeauthoryear{Al-Rfou \bgroup \em et al.\egroup
  }{2013}]{polyglot:2013:ACL-CoNLL}
Rami Al-Rfou, Bryan Perozzi, and Steven Skiena.
\newblock Polyglot: Distributed word representations for multilingual nlp.
\newblock In {\em Proceedings of the Seventeenth Conference on Computational
  Natural Language Learning}, pages 183--192, Sofia, Bulgaria, August 2013.
  Association for Computational Linguistics.

\bibitem[\protect\citeauthoryear{Bousquet and
  Elisseeff}{2002}]{bousquet2002stability}
Olivier Bousquet and Andr{\'e} Elisseeff.
\newblock Stability and generalization.
\newblock {\em Journal of machine learning research}, 2(Mar):499--526, 2002.

\bibitem[\protect\citeauthoryear{Bromley \bgroup \em et al.\egroup
  }{1993}]{bromley1993signature}
Jane Bromley, James~W Bentz, L{\'e}on Bottou, Isabelle Guyon, Yann LeCun, Cliff
  Moore, Eduard S{\"a}ckinger, and Roopak Shah.
\newblock Signature verification using a “siamese” time delay neural
  network.
\newblock {\em International Journal of Pattern Recognition and Artificial
  Intelligence}, 7(04):669--688, 1993.

\bibitem[\protect\citeauthoryear{Chapelle and Wu}{2010}]{chapelle2010gradient}
Olivier Chapelle and Mingrui Wu.
\newblock Gradient descent optimization of smoothed information retrieval
  metrics.
\newblock {\em Information retrieval}, 13(3):216--235, 2010.

\bibitem[\protect\citeauthoryear{Chapelle \bgroup \em et al.\egroup
  }{2011}]{chapelle2011future}
Olivier Chapelle, Yi~Chang, and Tie-Yan Liu.
\newblock Future directions in learning to rank.
\newblock In {\em Proceedings of the Learning to Rank Challenge}, pages
  91--100, 2011.

\bibitem[\protect\citeauthoryear{Deerwester \bgroup \em et al.\egroup
  }{1990}]{deerwester1990indexing}
Scott Deerwester, Susan~T Dumais, George~W Furnas, Thomas~K Landauer, and
  Richard Harshman.
\newblock Indexing by latent semantic analysis.
\newblock {\em Journal of the American society for information science},
  41(6):391--407, 1990.

\bibitem[\protect\citeauthoryear{Hardt \bgroup \em et al.\egroup
  }{2015}]{hardt2015train}
Moritz Hardt, Benjamin Recht, and Yoram Singer.
\newblock Train faster, generalize better: Stability of stochastic gradient
  descent.
\newblock {\em arXiv preprint arXiv:1509.01240}, 2015.

\bibitem[\protect\citeauthoryear{Huang \bgroup \em et al.\egroup
  }{2013}]{Huang2013}
Po-Sen Huang, Xiaodong He, Jianfeng Gao, Li~Deng, Alex Acero, and Larry Heck.
\newblock {Learning deep structured semantic models for web search using
  clickthrough data}.
\newblock In {\em Proc. of the Conference on Information and Knowledge
  Management (CIKM)}, 2013.

\bibitem[\protect\citeauthoryear{Lan \bgroup \em et al.\egroup
  }{2008}]{lan2008query}
Yanyan Lan, Tie-Yan Liu, Tao Qin, Zhiming Ma, and Hang Li.
\newblock Query-level stability and generalization in learning to rank.
\newblock In {\em Proceedings of the 25th international conference on Machine
  learning}, pages 512--519, 2008.

\bibitem[\protect\citeauthoryear{Lan \bgroup \em et al.\egroup
  }{2009}]{lan2009generalization}
Yanyan Lan, Tie-Yan Liu, Zhiming Ma, and Hang Li.
\newblock Generalization analysis of listwise learning-to-rank algorithms.
\newblock In {\em Proceedings of the 26th Annual International Conference on
  Machine Learning}, pages 577--584, 2009.

\bibitem[\protect\citeauthoryear{McCullagh}{1980}]{mccullagh1980regression}
Peter McCullagh.
\newblock Regression models for ordinal data.
\newblock {\em Journal of the Royal Statistical Society: Series B
  (Methodological)}, 42(2):109--127, 1980.

\bibitem[\protect\citeauthoryear{Mikolov \bgroup \em et al.\egroup
  }{2013}]{Mikolov2013}
Tomas Mikolov, Kai Chen, Greg Corrado, and Jeffrey Dean.
\newblock Distributed representations of words and phrases and their
  compositionality.
\newblock {\em Proc. of the Conference on Advances in Neural Information
  Processing Systems (NIPS)}, pages 1--9, 2013.

\bibitem[\protect\citeauthoryear{Nie}{2010}]{nie2010cross}
Jian-Yun Nie.
\newblock Cross-language information retrieval.
\newblock {\em Synthesis Lectures on Human Language Technologies}, 3(1):1--125,
  2010.

\bibitem[\protect\citeauthoryear{Nigam \bgroup \em et al.\egroup
  }{2019}]{Nigam2019}
Priyanka Nigam, Yiwei Song, Vijai Mohan, Vihan Lakshman, Weitian Ding, Choon
  {Hui Teo}, Hao Gu, Bing Yin, and Ankit Shingavi.
\newblock {Semantic Product Search}.
\newblock In {\em Proc. of the ACM SIGKDD Conference on Knowledge Discovery and
  Data Mining (KDD)}, 2019.

\bibitem[\protect\citeauthoryear{Palangi \bgroup \em et al.\egroup
  }{2016}]{Palangi2016}
Hamid Palangi, Li~Deng, Yelong Shen, Jianfeng Gao, Xiaodong He, Jianshu Chen,
  Xinying Song, and Rabab Ward.
\newblock {Deep Sentence embedding using long short-term memory networks:
  Analysis and application to information retrieval}.
\newblock {\em IEEE/ACM Transactions on Audio Speech and Language Processing},
  24(4):694--707, apr 2016.

\bibitem[\protect\citeauthoryear{Pennington \bgroup \em et al.\egroup
  }{2014}]{pennington-etal-2014-glove}
Jeffrey Pennington, Richard Socher, and Christopher Manning.
\newblock {G}lo{V}e: Global vectors for word representation.
\newblock In {\em Proc. of the Conference on Empirical Methods for Natural
  Language Processing (EMNLP)}, 2014.

\bibitem[\protect\citeauthoryear{Rennie and Srebro}{2005}]{rennie2005loss}
Jason~DM Rennie and Nathan Srebro.
\newblock Loss functions for preference levels: Regression with discrete
  ordered labels.
\newblock In {\em Proceedings of the IJCAI multidisciplinary workshop on
  advances in preference handling}, volume~1. Kluwer Norwell, MA, 2005.

\bibitem[\protect\citeauthoryear{Salakhutdinov and
  Hinton}{2009}]{Salakhutdinov2009}
Ruslan Salakhutdinov and Geoffrey Hinton.
\newblock {Semantic hashing}.
\newblock {\em International Journal of Approximate Reasoning}, 50(7):969--978,
  jul 2009.

\bibitem[\protect\citeauthoryear{Sasaki \bgroup \em et al.\egroup
  }{2018}]{Sasaki2018}
Shota Sasaki, Shuo Sun, Shigehiko Schamoni, Kevin Duh, and Kentaro Inui.
\newblock {Cross-Lingual Learning-to-Rank with Shared Representations}.
\newblock In {\em Proc. of the Annual Meeting of the North American Association
  of Computational Linguistics (NAACL)}, 2018.

\bibitem[\protect\citeauthoryear{Shen \bgroup \em et al.\egroup
  }{2014}]{Shen2014a}
Yelong Shen, Xiaodong He, Jianfeng Gao, Li~Deng, and Gr{\'{e}}goire Mesnil.
\newblock {A latent semantic model with convolutional-pooling structure for
  information retrieval}.
\newblock In {\em Proc. of the Conference on Information and Knowledge
  Management (CIKM)}, 2014.

\bibitem[\protect\citeauthoryear{Srivastava \bgroup \em et al.\egroup
  }{2014}]{srivastava2014dropout}
Nitish Srivastava, Geoffrey Hinton, Alex Krizhevsky, Ilya Sutskever, and Ruslan
  Salakhutdinov.
\newblock Dropout: a simple way to prevent neural networks from overfitting.
\newblock {\em The journal of machine learning research}, 15(1):1929--1958,
  2014.

\bibitem[\protect\citeauthoryear{Tewari and
  Chaudhuri}{2015}]{tewari2015generalization}
Ambuj Tewari and Sougata Chaudhuri.
\newblock Generalization error bounds for learning to rank: Does the length of
  document lists matter?
\newblock In {\em International Conference on Machine Learning}, pages
  315--323, 2015.

\bibitem[\protect\citeauthoryear{Zhou \bgroup \em et al.\egroup
  }{2012}]{zhou2012translation}
Dong Zhou, Mark Truran, Tim Brailsford, Vincent Wade, and Helen Ashman.
\newblock Translation techniques in cross-language information retrieval.
\newblock {\em ACM Computing Surveys (CSUR)}, 45(1):1--44, 2012.

\end{thebibliography}

\clearpage
\appendix

\section{Appendices}
\label{sec:appendix}
\renewcommand{\thesection}{A}
\setcounter{subsection}{0}
\renewcommand{\thesubsection}{A.\arabic{subsection}}
\setcounter{thm}{0}
\renewcommand{\thethm}{A.\arabic{thm}}

\subsection{Proof of Lemma \ref{lemma:smooth}}

\begin{definition}
A function $f: \Omega \to \mathbbm{R}$ is $L$-Lipschitz if for all $u, v \in \Omega$, we have
\begin{align*}
    |f(u) - f(v)| \leq L\|u-v\|.
\end{align*}
\end{definition}

\begin{lemma}\label{lemma:1}
If function $f$ and $g$ is $L_1$-Lipschitz and $L_2$-Lipschitz, then $g \circ f$ is $L_1L_2$-Lipschitz.
\end{lemma}
Here we use $\circ$ to be the function composition opertor, i.e. $g \circ f(\cdot) = g(f(\cdot))$. It is easy to prove the Lemma by the definition of $L$-Lipschitz. 

\begin{definition}
A function $f: \Omega \to \mathbbm{R}$ is $\beta$-smooth if for all $u, v \in \Omega$, we have
\begin{align*}
    \|\nabla f(u) - \nabla f(v)\| \leq \beta \|u-v\|.
\end{align*}
\end{definition}

\begin{lemma}\label{lemma:2}
If function $f$ is $L_1$-Lipschitz and $\beta_1$-smooth, and $g$ is $L_2$-Lipschitz and $\beta_2$-smooth, then $g \circ f$ is $(L_1 \beta_2 + L_2 \beta_1)$-smooth.
\end{lemma}
\begin{proof}
We have the following inequality,
\begin{align*}
   & \|\nabla g \circ f(u) - \nabla g \circ f(v)\|\\
   = & \|\nabla g(f(u)) \nabla f(u) - \nabla g(f(v)) \nabla f(v) \| \\
     \leq &\|\nabla g(f(u)) \nabla f(u) - \nabla g(f(u)) \nabla f(v) \| + \\
    & \|\nabla g(f(u)) \nabla f(v) - \nabla g(f(v)) \nabla f(v) \| \\
    \leq & L_2 \beta_1 \|u - v\| + L_1 \beta_2 \|u-v\|.
\end{align*}
In the second inequality, we use the Lipschitz property that $\nabla g(f(u)) \leq L_2$ and $\nabla f(v) \leq L_1$.
\end{proof}

\paragraph{Proof of Proposition \ref{prop:3}}
\begin{proof}
Let $K$ be the number of classes and $-1 = \theta_0 \leq \theta_1\leq ... \leq \theta_K = 1$. In Immediate Threshold with $L_2$ loss, 
\begin{align*}
    SOSL(r, y) = f_1(\theta_y - r) + f_1(r - \theta_{y-1}),
\end{align*}
where 
\[ f_1(x) = \begin{cases} 
      x^2 & x < 0 \\
      0 & x \geq 0
   \end{cases}.
\]
It is easy to see that 
\begin{align*}
    |\nabla SOSL(r, y)| \leq \max(2|1 + \theta_{y-1}|, 2|1-\theta_y|),
\end{align*}
for a given $y$ or
\begin{align*}
   & |\nabla SOSL(r, y)| 
    \leq \max_y(\max(2|1 + \theta_{y-1}|, 2|1-\theta_y|)) 
    \leq 4,
\end{align*}
for any $y$. Also 

\[ \nabla SOSL^2(r, y) = \begin{cases} 
      0 & r \in (\theta_{y-1}, \theta_y) \\
      2 & otherwise
   \end{cases}.
\]
for any $r$ and $y$, which means $|\nabla SOSL^2(r, y)| \leq 2$

\end{proof}

\paragraph{Proof of Lemma \ref{lemma:smooth}}
To have a high-level view of the proof, we decompose the loss function into three functions, and will show that each function is Lipschitz continuous and smooth with respect to its domain. 

\begin{proof}

Let $r(f) = \frac{f^T f_d}{(\|f\| + \epsilon)(\|f_d\|+\epsilon)}$ be a function of $f \in \mathbbm{R}^p$ when  $f_d \in \mathbbm{R}^p$ is fixed. For a fixed $\epsilon > 0$, the gradient of $\nabla r(f)$ exists, where 
\begin{align*}
    \nabla r(f) = & \frac{f_d^T}{(\|f\| + \epsilon)(\|f_d\|+\epsilon)} 
    - r(f)\frac{f^T }{(\|f\| + \epsilon)\|f\|} \\
   \|\nabla r(f) \| \leq & \Big{\|} \frac{f_d^T}{(\|f\| + \epsilon)(\|f_d\|+\epsilon)} \Big{\|} 
   + \Big{\|} r(f)\frac{f^T }{(\|f\| + \epsilon)\|f\|} \Big{\|}
    \leq   \frac{2}{(\|f\| + \epsilon)}  \leq  \frac{2}{\epsilon}
\end{align*}
To show $r(f)$ is also smooth, we calculate the Hessian matrix of $r(f)$. Let $M = (\|f\| + \epsilon)(\|f_d\|+\epsilon)$,
\begin{align*}
    \nabla^2 r(f) = &-\frac{f_d f^T}{M (\|f\| + \epsilon)\|f\|} 
    - \nabla r(f)^T \frac{f^T }{(\|f\| + \epsilon)\|f\|} 
     -r(f) \Big{[} \frac{\mathbbm{1}}{(\|f\| + \epsilon)\|f\|} - \frac{f (2f^T  + \epsilon \frac{f^T}{\|f\|})}{(\|f\| + \epsilon)^2\|f\|^2} \Big{]},
\end{align*}
where $\mathbbm{1}$ is a $p \times p$ identity matrix. It is easy to see that 
\begin{align*}
    \|\nabla^2 r(f)\| \leq \frac{N(p)}{\epsilon^2},
\end{align*}
where $N(p)$ is a constant only depends on $p$. Therefore, $r(f)$ is $2/\epsilon$-Lipschitz and $N(p)/\epsilon^2$-smooth. Similarly, this is also true with respect to $f_d$.

Deep Neural Networks normally do not enjoy Lipschitz continuity or smoothness due to their expressive power, which make it difficult to analyze its theoretical performance. We now show that our particular models are Lipschitz continuous and smooth for both query and document with respect to Embedding parameters. 

For a text query $Q$ with token length $t$, we represent it by $Q = (q_1, ..., q_t) \in \mathbbm{R}^{n_q \times t}$ where $q_i \in \mathbbm{R}^{n_q}$ be the one-hot vector with the index of $i$-th token in text $Q$ be 1 and the rest are 0. The output of query encoding networks $f_q$ is then

\begin{align*}
    f_q(Q) = &tanh(AvePooling(W_q^T Q)) \\
       = & tanh(\sum_{i=1}^t \frac{1}{t}W_q^T q_i)
\end{align*}
Note that hyperbolic tangent $tanh(x) = \frac{e^x - e^{-x}}{e^x + e^{-x}}$ is smooth, and $\sum_{i=1}^t \frac{1}{t}W_q^T q_i$ is linear in terms of $W_q$. We can claim that $f_q(Q)$ is smooth and Lipschitz continuous with respect to $W_q$.

Using Lemma \ref{lemma:1}, Lemma \ref{lemma:2} and Proposition \ref{prop:3}, we can prove Lemma \ref{lemma:smooth}.
\end{proof}

\subsection{Proof of Theorem \ref{thm}}
In the Ordinal Regression stage, let $x = (q, d)$ represent input query-document pair and let $y \in [K] = \{1, ..., K \}$ where $K$ is the number of relevance level. Given a sequence of samples $S = ((x_1, y_1), ..., (x_n, y_n)) \in (X \times [K])^n$, the goal is to learn a mapping $f : X \to [-1, 1]$ and a set of threshold $-1 = \theta^0 \leq \theta^1 \leq ... \leq \theta^{K-1} \leq \theta^K = 1$, then predicting the label based on the relevance score. Let $\theta = (\theta_1, ..., \theta_{K-1})$ and $E$ be a distribution on$(X, [K])$, also let $l(f, \theta, (x, y))$ be the error of a sample pair $(x, y)$ for a function $f$ and thresholds $\theta$, we define the expected error and empirical error as follows:
\begin{align}
    L^l_E(f, \theta) = \mathbf{E}_{(x, y)\sim E}[l(f, \theta, (x, y))]; \\
    L^l_S(f, \theta) = \frac{1}{m}\sum_{i=1}^m [l(f, \theta, (x_i, y_i))] .
\end{align}
We can prove the following Theorem from \cite{agarwal2008generalization}. 

\begin{thm}\label{thm:gen}
Let $\mathcal{A}$ be an ordinal regression algorithm which, given as input a training sample $S \in (X \times [K])^n$, learns a real-valued function $f_S: X \to \mathbbm{R}$ and a threshold vector $\theta_S \equiv (\theta_S^1, ..., \theta_S^{K-1})$. Let $l$ be any loss function in this setting such that $0 \leq l(f_S, \theta_S, (x, y)) \leq M$ for all training samples $S$ and all $(x, y) \in X \times [K]$, and let $\beta : \mathbbm{N} \to \mathbbm{R}$ be such that $\mathcal{A}$ has loss stability $\beta$ with respect to $l$. Then for any $0 < \delta < 1$ and for any distribution $D$ on $X \times [K]$, with probability at least $1- \delta$ over the draw of $S$, 
\begin{align*}
    &L^l_E(f_S, \theta_S) - L^l_S(f_S, \theta_S) \leq 
     2\beta(n) + (4n\beta(n) + M) \sqrt{\frac{\ln 1/ \delta}{2n}}
\end{align*}
\end{thm}
Note that the original theorem in \cite{agarwal2008generalization} is defined with $\theta$ over the real line $\mathbbm{R}$, but it is trivial to apply to $\theta$ defined over $[-1, 1]$. It is also easy to verify that $0 \leq l(f_S, \theta_S, (x, y)) \leq M$ since our Ordinal Regression loss is bounded. If $\beta(n)$ has a rate of $1/n$, then the generalization bound $L^l_E(f_S, \theta_S) - L^l_S(f_S, \theta_S)$ goes to zero as $n$ goes to infinity. Next, we will give the definition of stability $\beta$ and to show that DNNs optimized by SGD satisfy this requirement.  

\subsection{Stability}
The Uniform Stability is a measurement of how the algorithm will be affected if removing one sample from the training set. Let $S = ((x_1, y_1), ..., (x_n, y_n))$ be the training set and $S^{\setminus i}$ represents the set by removing the $i$-th element from $S$  
\begin{align*}
    S^{\setminus i} = &((x_1, y_1), ..., (x_{i-1}, y_{i-1}), 
    (x_{i+1}, y_{i+1}), ..., (x_m, y_m)).
\end{align*}The formal definition is first proposed by \cite{bousquet2002stability} and stated as follows,
\begin{definition}{(Uniform Stability)}
An algorithm $\mathcal{A}$ has uniform stability $\beta$ with respect to the loss function $l$ if the following holds
\begin{align*}
    &\forall S \in X^m, \forall i \in \{1, ..., m\},
    \|l(A_S, \cdot) - l(A_{S^{\setminus i}}, \cdot)\|_{\infty} \leq \beta.
\end{align*}
\end{definition}

Remember that our goal is to learn a function $f: X \to [-1, 1]$ such that $f(x) = \textit{cos}(f_q(q), f_d(d)) \in [-1, 1]$ which is parameterized by DNNs, and optimized by SGD. The following theorem (\cite{hardt2015train}, Theorem 3.8) shows the stability bound for convex loss minimization via SGD,

\begin{thm}\label{thm:sta}
Assume that the loss function $l(\cdot; y) \in [0, 1]$ is $\gamma$-smooth and $L$-Lipschitz for every $y$. Suppose that we run SGD for $T$ steps with monotonically non-increasing step sizes $\alpha_t \leq c/t$. Then SGD satisfies uniform stability with 
\begin{align*}
    \beta_{stab} \leq \frac{1 + 1/\gamma c}{n-1}(2cL^2)^{\frac{1}{\gamma c + 1}}T^{\frac{\gamma c}{\gamma c +1}}.
\end{align*}
\end{thm}

\paragraph{Proof of Theorem \ref{thm}}
By combining Lemma \ref{lemma:smooth}, Theorem \ref{thm:gen} and Theorem \ref{thm:sta}, we can prove Theorem \ref{thm}.

\end{document}